\PassOptionsToPackage{dvipsnames}{xcolor}
\PassOptionsToPackage{nosort,nocompress}{cite}
\documentclass[a4paper,colorlinks]{preprint}
\usepackage[english]{babel}
\usepackage{mathtools}
\numberwithin{equation}{section}
\usepackage{amsthm}
\usepackage{thmtools}
\usepackage{amssymb}
\usepackage{bm}
\usepackage{amsfonts}
\usepackage{enumitem}
\setlist[itemize]{noitemsep, topsep=3pt}
\usepackage{microtype}
\usepackage{tikz}
\usetikzlibrary{arrows.meta,calc,decorations.pathreplacing}
\usepackage{complexity}
\usepackage{xspace}
\usepackage{comment}
\usepackage{bbm}
\usepackage[nameinlink,capitalize]{cleveref}

\AddToHook{cmd/oldmaketitle/before}{%
  \ihead*{}\chead*{}\ohead*{}
  \ifoot*{}\ofoot*{}
  \cfoot*{\usekomafont{pagenumber}\thepage}
  \KOMAoptions{headsepline=false,plainheadsepline=false,
    footsepline=false,plainfootsepline=false}
}

\newtheorem{theorem}{Theorem}[section]
\newtheorem{lemma}[theorem]{Lemma}
\newtheorem{corollary}[theorem]{Corollary}
\newtheorem{definition}[theorem]{Definition}
\newtheorem{remark}[theorem]{Remark}

\crefname{appendix}{Appendix}{Appendices}
\Crefname{appendix}{Appendix}{Appendices}
\crefname{remark}{remark}{remarks}
\Crefname{remark}{Remark}{Remarks}

\newcommand{\NN}{\mathbb{N}}
\newcommand{\CC}{\mathbb{C}}

\newcommand{\FF}{\mathbb{F}}
\newcommand{\ZZ}{\mathbb{Z}}
\newcommand{\defeq}{\coloneqq}
\newcommand{\bl}{\boldsymbol{\lambda}}

\newcommand{\elem}{\operatorname{elem}}
\newcommand{\pow}{\operatorname{pow}}
\newcommand{\Per}{\operatorname{Per}}
\newcommand{\rPer}{\operatorname{rPer}}

\newcommand{\Hom}{\operatorname{Hom}}
\newcommand{\card}[1]{\left|#1\right|}

\newcommand{\paren}[1]{\left(#1\right)}
\newcommand{\norm}[1]{\left\|#1\right\|}
\newcommand{\PMP}{\operatorname{PerfMatch}}
\newcommand{\MatchP}{\operatorname{Match}}
\newcommand{\MPG}[1]{\MatchP\paren{#1}}
\newcommand{\PMPG}[1]{\PMP\paren{#1}}
\newcommand{\wt}[1]{\operatorname{wt}\paren{#1}}
\newcommand{\VBP}{\mathsf{VBP}}

\newcommand{\examplevertexdata}{1/-1.3/1.3/4/23,2/1.3/-1.3/3/14,3/1.3/1.3/3/106,4/-1.3/-1.3/2/109}
\newcommand{\exampleedgedata}{1/2/4/3,1/3/4/3,1/4/4/2,2/3/3/3,2/4/3/2,3/4/3/2}
\tikzset{
  paper figure/.style={font=\small},
  original vertex/.style={circle, draw, fill=white, inner sep=0pt, minimum size=5mm},
  copy vertex/.style={circle, fill=black, inner sep=0pt, minimum size=2.5pt},
  blob/.style={draw=none, fill=black!22},
  between/.style={draw=black!50, line width=.35pt},
  within/.style={draw=black, line width=.65pt},
  figure label/.style={fill=white, inner sep=1pt},
  part token/.style={circle, draw=black!75, line width=.35pt, fill=black!22, inner sep=0pt, minimum size=6mm}
}
\newcommand{\drawexampleblowup}[2]{%
  \foreach \i/\x/\y/\w/\angle in \examplevertexdata {
    \coordinate (#1-\i) at (\x,\y);
    \draw[blob] (#1-\i) circle[radius=.48cm];
    \foreach \j in {1,...,\w}
      \coordinate (#1-\i-\j) at
        ($(#1-\i)+({\angle+360*(\j-1)/\w}:.29)$);
  }
  \foreach \i/\j/\wi/\wj in \exampleedgedata
    \foreach \a in {1,...,\wi}
      \foreach \b in {1,...,\wj}
        \draw[between] (#1-\i-\a) -- (#1-\j-\b);
  \ifnum#2=1\relax
    \foreach \i/\x/\y/\w/\angle in \examplevertexdata {
      \foreach \a in {1,...,\w} {
        \foreach \b in {1,...,\w} {
          \ifnum\a<\b
            \draw[within] (#1-\i-\a) -- (#1-\i-\b);
          \fi
        }
      }
    }
  \fi
  \drawexamplevertices{#1}
}
\newcommand{\drawexamplevertices}[1]{%
  \foreach \i/\x/\y/\w/\angle in \examplevertexdata
    \foreach \a in {1,...,\w}
      \node[copy vertex] at (#1-\i-\a) {};
}
\newcommand{\labelexampleblowup}[1]{%
  \foreach \i in {1,4}
    \node[left=5mm] at (#1-\i) {$\i$};
  \foreach \i in {2,3}
    \node[right=5mm] at (#1-\i) {$\i$};
}

\title{When Are Monomial Symmetric Polynomials~Hard?}
\author[1]{Markus~Bl{\"a}ser}
\author[2]{Radu~Curticapean\thanks{Funded by the European Union (ERC, CountHom, 101077083). Views and opinions expressed are however those of the author(s) only and do not necessarily reflect those of the European Union or the European Research Council Executive Agency.}}
\author[3]{Sagnik~Dutta}
\author[4]{Gorav~Jindal\protect\footnotemark[1]}
\affil[1]{Saarland University, Saarland Informatics Campus, Saarbr\"ucken, Germany\authorcr
  \email{mblaeser@cs.uni-saarland.de}}
\affil[2]{University of Regensburg, Regensburg, Germany, and IT University of Copenhagen, Copenhagen, Denmark\authorcr
  \email{radu.curticapean@ur.de}}
\affil[3]{Max Planck Institute for Informatics, Saarland Informatics Campus, Saarbr\"ucken, Germany\authorcr
  \email{sadutta@mpi-inf.mpg.de}}
\affil[4]{University of Regensburg, Regensburg, Germany\authorcr
  \email{gorav.jindal@ur.de}}

\shorttitle{Monomial Symmetric Polynomials}
\shortauthor{Bl{\"a}ser et al.}
\shortdate{~}
\novelty{}
\date{}

\begin{document}

\abstract{
For a partition $\boldsymbol{\lambda} = (\lambda_1,\ldots,\lambda_r)$ and $n\ge r$, the \emph{monomial symmetric polynomial} is defined as 
$
m_{\boldsymbol{\lambda}}(x_1,\ldots,x_n)
:=
\sum x_{i_1}^{\lambda_1} x_{i_2}^{\lambda_2} \cdots x_{i_r}^{\lambda_r},
$
where the sum ranges over all distinct monomials obtained by assigning the parts of $\boldsymbol{\lambda}$ to pairwise distinct variables.
We establish an almost-complete dichotomy for the complexity of the monomial symmetric polynomials $m_{\boldsymbol{\lambda}}$:
If a family of partitions contains polynomially many distinct parts, then the corresponding monomial symmetric polynomials (on sufficiently many variables) are $\mathsf{VNP}$-complete.
If the number of distinct parts is constant, the family admits polynomial-size circuits.
Thus the only remaining gap concerns the intermediate regime where the number of distinct parts grows subpolynomially.

As part of our proof, we use combinatorial properties of matching polynomials to establish $\mathsf{VNP}$-hardness of the perfect matching polynomial for \emph{blow-up graphs}; these are graphs with non-generic edge-weights that are obtained by ``blowing up'' a smaller graph.
We reduce these matching polynomials to monomial-symmetric polynomials by combining classical techniques from algebraic complexity theory with structural properties of Sidon sets.
}

\maketitle
\hypersetup{
  pdfauthor={Markus Bl\"aser, Radu Curticapean, Sagnik Dutta, Gorav Jindal}
}

\section{Introduction.}
Symmetry is a fundamental concept in mathematics and computer science, appearing across algebra, combinatorics, and theoretical computer science. In algebra, it appears prominently with the class of \emph{symmetric polynomials}, i.e., polynomials which are invariant under any permutation of their variables. Symmetric polynomials have been extensively studied in algebraic combinatorics and invariant theory, and they naturally arise in counting problems and algebraic computation. Understanding the computational complexity of natural families of symmetric polynomials is therefore a central question in algebraic complexity theory.

In algebraic complexity theory, the underlying complexity measure is usually the size of \emph{arithmetic circuits} computing the polynomial. An arithmetic circuit is a directed acyclic graph whose internal nodes are labeled by addition or multiplication gates, and whose leaves are labeled by input variables or field constants. The arithmetic complexity $L(f)$ of a polynomial $f\in \FF[x_1,\ldots,x_n]$ is the minimum size of an arithmetic circuit computing it, where circuit size is the number of nodes. We refer the reader to the surveys \cite{mahajan2013algebraiccomplexityclasses, shpilka2010arithmeticcircuitssurvey,saptharishi2015survey,burgisser2016permanentversusdeterminant} for background on algebraic complexity theory.

\paragraph{Monomial symmetric polynomials.}
In this paper, we investigate the circuit complexity of \emph{monomial symmetric polynomials}. Let $\bl=(\lambda_1,\ldots,\lambda_r)$ be a partition with $\lambda_1 \ge \cdots \ge \lambda_r > 0$. The monomial symmetric polynomial $m_{\bl}$ associated with $\bl$ in $n$ variables is 
\[
m_{\bl}(x_1,\ldots,x_n)
\defeq
\sum x_{i_1}^{\lambda_1}x_{i_2}^{\lambda_2}\cdots x_{i_r}^{\lambda_r},
\]
where the sum ranges over all \emph{distinct} monomials obtained by assigning the parts of $\bl$ to pairwise distinct variables; equivalently, we sum over the $S_n$-orbit but keep each resulting monomial only once. 

The family of monomial symmetric polynomials forms one of the standard bases of the ring of symmetric polynomials and plays a fundamental role in algebraic combinatorics and representation theory. 
Classical symmetric polynomials, such as elementary and power symmetric polynomials, can be viewed as special cases of monomial symmetric polynomials for specific partitions. For instance, the $k$-th elementary symmetric polynomial is $e_k=m_{(1,1,\ldots,1)}$, while the $k$-th power sum symmetric polynomial is $p_k = m_{(k)}$. 

\subsection{Complexity of Symmetric Polynomials.}

Symmetric polynomials have been studied in algebraic complexity theory for decades.
As a first question, one can ask whether there exist \emph{any} families of symmetric polynomials that require super-polynomial-size arithmetic circuits. Since we lack superpolynomial circuit lower bounds for explicit families, these results either need to address restricted models or require complexity-theoretic assumptions.

\paragraph{Unconditional lower bounds.} A seminal result of Nisan and Wigderson~\cite{NisanWigderson1996} established that elementary symmetric polynomials $e_k(x_1,\ldots,x_n)$ require exponential-size homogeneous depth-three circuits over fields of characteristic zero when $k$ is even, $k=\Theta(n)$, and $k\le n/3$.
For non-homogeneous depth-three formulas, Shpilka and Wigderson~\cite{shpilkawigdersondepth3arithmetic1999} prove a quadratic lower bound for computing $e_k(x_1,\ldots,x_n)$ when $k=\Theta(n)$ and $k\le 2n/3$, which is tight. The upper bound is due to Ben-Or, see \cite[Thm 4.2]{shpilkawigdersondepth3arithmetic1999}. Hrube{\v{s}} and Yehudayoff \cite{HrubesYehudayoff2011} investigate formulas of arbitrary depth. They prove that, for $n\ge 2k$, every multilinear homogeneous formula for $e_k(x_1,\ldots,x_n)$ has size $k^{\Omega(\log k)} n$ and construct homogeneous formulas of matching size, which are, however, not multilinear. Fournier \emph{et al.}\ \cite{fournier2017shiftedpartialderivativecomplexityelementarysymmetricpolynomials} study the shifted partial derivative complexity of elementary symmetric polynomials, strengthening the results of Nisan and Wigderson in certain parameter regimes.
Despite this rich history, the complexity of monomial symmetric polynomials has remained largely open.

\paragraph{Conditional lower bounds.}
Bl{\"a}ser and Jindal \cite{blasercomplexitysymmetricpolynomials2019} constructed families of symmetric polynomials that require circuits of superpolynomial size assuming $\VP \not= \VNP$. 
More precisely, their families are $\VNP$-complete under $c$-reductions, which implies they cannot be computed by polynomial-size arithmetic circuits unless $\VP = \VNP$. Here, $c$-reductions are the algebraic analogue of Turing reductions in classical complexity theory, allowing for oracle access to the polynomial being reduced to, see \Cref{sec:background} and \Cref{sec:preliminaries} for a formal definition.

The main tool of Bl{\"a}ser and Jindal is an effective version of the fundamental theorem for symmetric polynomials (\Cref{thm:symmetric-basis}) in terms of circuit complexity: 
If $(f_n)$ is a family of $\VNP$-complete polynomials, then the polynomial family $g_n$ defined by $g_n \defeq f_n(e_1, e_2, \dots, e_{v(n)})$ is a family of symmetric polynomials that is $\VNP$-complete under $c$-reductions, where $e_i$ is the $i$-th elementary symmetric polynomial and $v(n)$ is the number of variables in $f_n$. 
Chaugule \emph{et al.}~\cite{chaugule2023schurpolynomialsdonothavesmallformulasifdeterminantdoesnt} prove a similar result for formulas, and Bhattacharjee \emph{et al.}~\cite{bhattacharjeegcd2025} prove a constant-depth analogue.

Similar techniques also apply to Schur polynomials, another basis family of symmetric polynomials that play a central role in algebraic combinatorics and representation theory. 
They can be defined as specific determinants, and Chaugule \emph{et al.}~\cite{chaugule2023schurpolynomialsdonothavesmallformulasifdeterminantdoesnt} show that for certain families of partitions, Schur polynomials are hard for $\VBP$. Thus, if these specific determinants had algebraic formulas of polynomial size, then so would the ordinary determinants. 

\paragraph{Known results for monomial-symmetric polynomials.}

Curticapean \emph{et al.}~\cite{curticapeanmonomialmymmetricpolynomials2022} constructed an explicit family of monomial-symmetric polynomials that are $\VNP$-complete under $c$-reductions. Their hard family is specifically constructed to enable the hardness result. For example, the chosen partition has only distinct parts, and more importantly, these parts essentially enjoy the highly non-generic property of forming a \emph{Sidon set}, i.e., a set of integers $S$ such that $x,y\in S$ can be recovered uniquely from $x+y$.

The result in~\cite{curticapeanmonomialmymmetricpolynomials2022} was shown through a $c$-reduction from the $\VNP$-complete perfect matching polynomials to monomial symmetric polynomials. The starting point is the Binet--Minc identity (\Cref{thm:binet-minc}), which expresses monomial symmetric polynomials in terms of power sum symmetric polynomials as
\[
m_{\bl}(x_1, x_2, \ldots, x_n) = f(p_1, p_2, \ldots, p_d),
\]
for some ``witnessing polynomial'' $f$, where $d$ is the degree of $m_{\bl}$ and $p_i$ is the $i$-th power sum symmetric polynomial. By crucially using the Sidon properties of $\lambda$ and additional technical conditions on $\lambda$, a relatively direct reduction from perfect matching polynomials to this witnessing polynomial $f$ can be established.
Then, using results from \cite{chaugule2023schurpolynomialsdonothavesmallformulasifdeterminantdoesnt,blasercomplexitysymmetricpolynomials2019}, the witnessing polynomial $f(y_1, y_2, \ldots, y_d)$ is $c$-reduced to the monomial-symmetric polynomial $m_{\bl}$, giving the overall $\VNP$-hardness result.

\subsection{Our Results.}

We establish an almost-complete dichotomy for the algebraic complexity of monomial symmetric polynomials in terms of the number of distinct parts in the defining partition:
At one end, a \emph{constant} number of distinct parts leads to polynomial-size circuits.
At the other end, \emph{polynomially many} distinct parts lead to $\VNP$-completeness under $c$-reductions, provided that the number of variables in $m_\lambda$ is large enough.

\begin{theorem}
\label{thm:main-vnp-hardness}
Let $(f_n)_{n \in \NN}$ be a $p$-family of monomial symmetric polynomials such that the number of distinct parts in the defining partition of $f_n$ is at least $n^{\epsilon}$ for some $\epsilon>0$. Then $(f_n)_{n \in \NN}$ is $\VNP$-complete under $c$-reductions, assuming the number of variables in $f_n$ is at least the degree of $f_n$.
\end{theorem}

The tractable side of our classification is folklore; we state a bound for completeness in \Cref{sec:missing-proofs}.
\begin{theorem}
\label{thm:easy-case}
The monomial symmetric polynomial $m_{\bl}(x_1, x_2, \ldots, x_n)$ of degree $d$ can be computed by an arithmetic circuit of size at most $\poly(d)\cdot n^{O(t)}$ if the number of distinct parts in the partition $\bl$ is $t$.
\end{theorem}

Our classification has two limitations: First, \Cref{thm:main-vnp-hardness} requires that the number of variables is at least the degree of the polynomials in question. This is a genuine limitation. 
For example, the Vandermonde permanent discussed in \Cref{sec:conclusion-and-open-problems} has degree $\binom{n+1}{2}$ but only $n$ variables, making our theorem inapplicable.

Second, the classification leaves a blind spot in the intermediate regime where the number of distinct parts grows sub-polynomially in $n$.
A similar situation occurred in the classification of \emph{immanants}, another family of polynomials that are indexed by partitions $\lambda$.
B\"urgisser \cite{DBLP:books/daglib/0025071} proved the first hardness results for certain immanants, and Curticapean \cite{DBLP:conf/stoc/Curticapean21} later obtained a full dichotomy. Curiously, this dichotomy \emph{does} address regimes of sub-polynomially growing hardness parameters by using assumptions from parameterized complexity theory.
In our setting of monomial symmetric polynomials, both this issue and the overall hardness proof appear significantly more difficult: 
Immanants have matrix inputs, which directly give a very useful graph-theoretic structure that allows for a reduction from matchings, whereas hardness proofs for monomial symmetric polynomials have to recover this matching structure from additive relations among these exponents.
We return to this in \Cref{sec:conclusion-and-open-problems}.

\subsection{Proof Ideas.}

As in previous works~\cite{curticapeanmonomialmymmetricpolynomials2022}, we construct a $c$-reduction from the $\VNP$-complete perfect matching polynomials to monomial symmetric polynomials. We start again from the Binet--Minc identity (\Cref{thm:binet-minc}) that expresses $m_{\bl}$ of degree $d$ in terms of the power sum symmetric polynomials $p_i$, as sketched above for~\cite{curticapeanmonomialmymmetricpolynomials2022}. That is,
\[
m_{\bl}(x_1, x_2, \ldots, x_n) = f(p_1, p_2, \ldots, p_d).
\]
For arbitrary partitions $\lambda$ with many distinct parts, two new technical issues arise:
\begin{itemize}
    \item First, for general $\lambda$, we have no chance of controlling the additive relations among all distinct parts of $\bl$, so we must ensure carefully that we can extract a well-behaved subset. Indeed, we manage to extract a large sum-free Sidon subset from the distinct parts of $\lambda$.
    \item Second, repeated parts are a nontrivial issue that was invisible in the previous approaches. Repetitions essentially turn a generic graph into a \emph{blow-up graph}, where each selected part contributes a vertex block of size equal to the number of repetitions. Proving that the corresponding blow-up matching polynomials remain hard is a central new ingredient of the paper, and it is also a source of structured, non-generic permanent-like polynomial families.
\end{itemize}
In more detail, our proof strategy is as follows. Suppose that 
\[
\bl = (\lambda_1, \lambda_2, \ldots, \lambda_k)
= (\mu_1^{w_1}, \mu_2^{w_2}, \ldots, \mu_t^{w_t}),
\]
where $\mu_1 > \mu_2 > \cdots > \mu_t$ are the distinct parts of $\bl$, and $w_i$ is the multiplicity of $\mu_i$ in $\bl$.
The Binet--Minc identity establishes a correspondence between monomials of $f$ and set partitions of $[k]$. From the set $\{\mu_1,\ldots,\mu_t\}$ we extract a sum-free Sidon subset $S$ of size $\Omega(\sqrt t)$. The Sidon property makes pairwise sums of elements of $S$ uniquely decodable, and these sums will become edge variables of a complete graph.

We then take a carefully chosen homogeneous component of $f$ and set $y_{\mu}=0$ for every $\mu\in S$. This removes singleton blocks whose labels lie in $S$. Through a counting argument, we can force every surviving set partition to use all elements outside $S$ as singleton blocks, while the copies of elements in $S$ must be paired into blocks of size two; larger blocks would leave too few blocks in total. After dividing out the forced singleton contribution from the elements outside $S$, the remaining monomials are exactly encoded by pairings among the copies of elements of $S$.

If every selected part occurs once, these pairings are perfect matchings of a complete graph on $\Omega(\sqrt t)$ vertices, recovering essentially the situation from~\cite{curticapeanmonomialmymmetricpolynomials2022}. 
With multiplicities however, each selected part $\mu_i$ contributes $w_i$ indistinguishable copies, and the extracted polynomial is the perfect matching polynomial of a blow-up of the complete graph. The blow-up graph replaces vertex $i$ by a clique of size $w_i$; see \Cref{fig:partition-lemma} for an example and \Cref{sec:preliminaries} for the formal definition.

In \Cref{sec:complexity-of-blow-up-matching-polynomial}, we reduce the ordinary perfect matching polynomial to the perfect matching polynomial of a blow-up of a complete graph with arbitrary positive multiplicities and an even total number of vertices. The reduction attaches private gadget blocks and chooses algebraic edge-weights so that, after taking the right homogeneous component, every original block is used exactly once.

\subsection{Organization of the Paper.}

In \Cref{sec:preliminaries}, we introduce the necessary background and notation. 
In \Cref{sec:complexity-of-blow-up-matching-polynomial}, we reduce the perfect matching polynomial to its blown-up version for arbitrary positive multiplicities with an even total sum.
Finally, in \Cref{sec:blow-up-matching-polynomial-to-monomial-symmetric-polynomials}, we reduce the blown-up perfect matching polynomial to the monomial symmetric polynomial $m_{\bl}$.
Open problems are discussed in \Cref{sec:conclusion-and-open-problems}.
\Cref{sec:background} and \Cref{sec:missing-proofs} contain background material and omitted proofs.

\section{Preliminaries.}\label{sec:preliminaries}
We refer to \Cref{sec:apx-act} for a basic introduction to algebraic complexity. We will be working with $c$-reductions, which are the algebraic analogue of Turing reductions from classical complexity theory. The oracle complexity of a polynomial $f$ with respect to another polynomial $g$ is the minimum size of an arithmetic circuit computing $f$ that is allowed to query an oracle for $g$, denoted by $L^g(f)$. For more details on $c$-reductions, see \Cref{sec:apx-act}.

\subsection{Notation for Partitions.}
A partition $\bl$ of a positive integer $d$ is a tuple of positive integers $\bl=(\lambda_1,\lambda_2,\ldots,\lambda_k)$ such that $\lambda_1 \ge \lambda_2 \ge \cdots \ge \lambda_k > 0$ and $\sum_{i=1}^k \lambda_i = d$. We write $\bl \vdash d$. For a partition $\bl$, we define:
\begin{alignat*}{2}
d(\bl) &\defeq \sum_{i\in[k]} \lambda_i &\quad &\text{(degree)},\\
\ell(\bl) &\defeq k &\quad &\text{(number of parts)},\\
D(\bl) &\defeq \card{\{\lambda_1,\ldots,\lambda_k\}} &\quad &\text{(number of distinct parts)}.
\end{alignat*}
We also use a shorthand notation for partitions with repeated parts. If $\mu_1 > \mu_2 > \cdots > \mu_t$ (with $t = D(\bl)$) are the distinct parts of $\bl$ with multiplicities $w_1, w_2, \ldots, w_t$ respectively, we write
\[
\bl = (\mu_1^{w_1}, \mu_2^{w_2}, \ldots, \mu_t^{w_t}).
\]

\subsection{The Ring of Symmetric Polynomials.}
Unless stated otherwise, we assume that the underlying field $\FF$ is $\CC$. 
A polynomial $f \in \FF[x_1, x_2, \ldots, x_n]$ is called \emph{symmetric} if, for every permutation 
$\sigma \in S_n$, we have
\[
f(x_1, x_2, \ldots, x_n) = f(x_{\sigma(1)}, x_{\sigma(2)}, \ldots, x_{\sigma(n)}).
\]
The ring of symmetric polynomials in $n$ variables admits two natural sets of algebraically independent generators: the \emph{elementary symmetric polynomials} $e_1,\ldots,e_n$ and the \emph{power sum symmetric polynomials} $p_1,\ldots,p_n$, defined for $k\in[n]$ by
\begin{equation*}
e_k(x_1,\ldots,x_n) \defeq \sum_{1 \le i_1 < i_2 < \cdots < i_k \le n} x_{i_1} x_{i_2} \cdots x_{i_k}
\qquad \text{and} \qquad 
p_k(x_1,\ldots,x_n) \defeq \sum_{i=1}^{n} x_i^k.
\end{equation*}
Every symmetric polynomial can be uniquely expressed as a polynomial in the elementary symmetric polynomials, and also as a polynomial in the power sum symmetric polynomials. See \Cref{sec:apx-symm-poly} for details.
For a symmetric polynomial $f$, we define $f^{\elem}$ and $f^{\pow}$ to be the unique polynomials in $\FF[y_1,\ldots,y_n]$ such that
\[ 
f = f^{\elem}(e_1, \ldots, e_n) \quad \text{and} \quad f = f^{\pow}(p_1, \ldots, p_n).
\]

\subsection{Matching Polynomials.}

Let $G$ be an $n$-vertex graph whose edges are labelled by variables from a set $\mathbf{X}$ or constants from $\FF$. For a matching $M$ of $G$, denote its weight by $\wt{M}$; this is the product of the weights of the edges in $M$.
The \emph{matching polynomial} of $G$, denoted $\mathrm{Match}(G) \in \FF[\mathbf{X}]$, is defined as
\[
\MPG{G} \defeq \sum_{M} \wt{M},
\]
where $M$ ranges over all matchings of $G$.
In a similar way, the \emph{perfect matching polynomial} of $G$ is defined by
\[
\PMPG{G} \defeq \sum_{M} \wt{M},
\]
where the sum ranges over all perfect matchings $M$ of $G$.

Assume that $n \in \NN$ is even. Let 
$\mathbf{X}_n \defeq \{x_{i,j} \mid 1 \le i < j \le n\}$
be the set of $\binom{n}{2}$ variables corresponding to the edges of the complete graph $K_n$. 
We consider the edge $(i,j)\in E(K_n)$ to be weighted by $x_{i,j}$. 
The $\VNP$-completeness of the permanent~\cite{valiantcompletenessclassesalgebra1979} implies $\VNP$-completeness of the polynomial
\[
\PMP_n \defeq \PMPG{K_n}.
\]
Indeed we have $\Per(A) = \PMPG{G}$ for the complete bipartite graph $G$ whose edge $(i,j)$ has weight $A_{i,j}$.

We will also consider the so-called \emph{blow-up} matching polynomials of complete graphs.
For any $n\in\NN$ and a vector $\mathbf{w} = (w_1, w_2, \ldots, w_n)$ of positive integers, the $\mathbf{w}$-\emph{blow-up} of $K_n$, denoted $K_n^{(\mathbf{w})}$, is the edge-weighted graph obtained by replacing each vertex $i$ with a clique of size $w_i$, and replacing each edge $\{i,j\}$ with a complete bipartite graph between the corresponding cliques. See \Cref{fig:graph-blow-up} for an example.
Every edge between blocks $i$ and $j$ has weight $x_{i,j}$. The edges inside block $i$ have weight $z_i$, where $z_i$ is a new variable.
Thus the underlying unweighted graph is complete on $\sum_i w_i$ vertices, but its edge weights are shared according to the blocks.

For a vector $\mathbf{w}\in \ZZ_{>0}^n$, we define the polynomials $\MatchP_n^{(\mathbf{w})}$ and $\PMP_n^{(\mathbf{w})}$ to be the matching polynomial and the perfect matching polynomial, respectively, of the blow-up graph $K_n^{(\mathbf{w})}$.
Throughout the paper, blow-up vectors have strictly positive integer entries. (A zero entry would simply delete the corresponding vertex and could be removed in advance.)

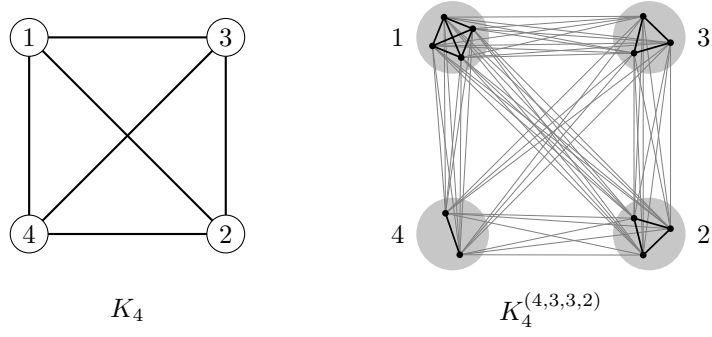
\begin{figure}[t]
\centering
\begin{tikzpicture}[paper figure]
\begin{scope}[xshift=-2.8cm]
  \foreach \i/\x/\y/\w/\angle in \examplevertexdata
    \coordinate (L-\i) at (\x,\y);
  \foreach \i/\j/\wi/\wj in \exampleedgedata
    \draw[thick] (L-\i) -- (L-\j);
  \foreach \i in {1,...,4}
    \node[original vertex] at (L-\i) {$\i$};
  \node at (0,-2.3) {$K_4$};
\end{scope}
\begin{scope}[xshift=2.8cm]
  \drawexampleblowup{R}{1}
  \labelexampleblowup{R}
  \node at (0,-2.3) {$K_4^{(4,3,3,2)}$};
\end{scope}
\end{tikzpicture}
\caption{The complete graph $K_4$ and its blow-up with $\mathbf w=(4,3,3,2)$.
Block $i$ contains a clique of size $w_i$ (black edges).
Each edge of $K_4$ is replaced by all edges between its two blobs (gray edges).
Edges within blob $i$ have weight $z_i$; edges between blobs $i$ and $j$ have weight $x_{i,j}$.}
\label{fig:graph-blow-up}
\end{figure}

\subsection{Monomial Symmetric Polynomials and Rectangular Permanents.}
\label{sec:monsympol}

Fix a $k \times n$ matrix $X$ where $k \le n$ and the $(i,j)$-th entry of $X$ is a variable $x_{i,j}$. For a positive integer $m$, we denote by $[m]$ the set $\{1,2,\ldots,m\}$. For a subset $J \subseteq [n]$ of size $k$, we denote by $X_J$ the submatrix obtained by keeping only the columns indexed by $J$.
The \emph{rectangular permanent} $\rPer_{k,n}$ is defined by
\[
\rPer_{k,n}(X) = \sum_{J \in \binom{[n]}{k}} \Per(X_J).
\]
The following polynomial identity will play a key role in our arguments.
For a positive integer $k$, let $\mathcal{P}_k$ denote the set of all partitions of $[k]$ into non-empty subsets.
\begin{theorem}[Binet--Minc Identity~\cite{minc1984permanents}]\label{thm:binet-minc}
Let $\FF$ be any field and let $X$ be a $k \times n$ matrix as above. 
Then
\[
\rPer_{k,n}(X) =
\sum_{\mathcal{I} \in \mathcal{P}_k}
(-1)^{k - |\mathcal{I}|}
\prod_{I \in \mathcal{I}} (|I| - 1)! \, S_I 
\quad\text{with}\quad 
S_I \defeq \sum_{j=1}^{n} \prod_{i \in I} x_{i,j}.
\]
\end{theorem}

Recall that the monomial symmetric polynomial $m_{\bl}(x_1,x_2,\ldots,x_n)$ associated with the partition $\bl$ in $n$ variables is defined as
\[
m_{\bl}(x_1,x_2,\ldots,x_n)
\defeq
\sum
x_{i_1}^{\lambda_1}
x_{i_2}^{\lambda_2}
\cdots
x_{i_k}^{\lambda_k}.
\]
Here the sum ranges over all distinct monomials obtained by assigning the parts of $\bl$ to pairwise distinct variables; equivalently, over all monomials whose nonzero exponent multiset is $\{\lambda_1,\ldots,\lambda_k\}$.
When the number of variables is clear from the context, we simply write $m_{\bl}$. 

Write $\bl=(\mu_1^{w_1},\ldots,\mu_t^{w_t})$, where the $\mu_j$ are distinct, and put $\kappa(\bl)\defeq\prod_{j=1}^t w_j!$. For $k\le n$, we have
\[
\rPer_{k,n}\!
\begin{pmatrix}
x_1^{\lambda_1} & x_2^{\lambda_1} & \cdots & x_n^{\lambda_1} \\
x_1^{\lambda_2} & x_2^{\lambda_2} & \cdots & x_n^{\lambda_2} \\
\vdots & \vdots & \ddots & \vdots \\
x_1^{\lambda_k} & x_2^{\lambda_k} & \cdots & x_n^{\lambda_k}
\end{pmatrix}
=\kappa(\bl)\,m_{\bl}(x_1,\ldots,x_n).
\]
Indeed, the rectangular permanent sums over injective assignments of rows to columns. Each distinct monomial occurs exactly $\kappa(\bl)$ times, since the $w_j$ rows with exponent $\mu_j$ can be assigned to their chosen columns in $w_j!$ ways.
For the matrix above, the block sums in \Cref{thm:binet-minc} are
\[
S_I=\sum_{j=1}^n x_j^{\lambda_I}=p_{\lambda_I}(x_1,\ldots,x_n),
\qquad \lambda_I\defeq\sum_{i\in I}\lambda_i.
\]
Dividing the Binet--Minc identity by $\kappa(\bl)$ therefore yields the following formula.

\begin{corollary}\label{cor:binet-minc}
Let $\bl=(\lambda_1,\ldots,\lambda_k)=(\mu_1^{w_1},\ldots,\mu_t^{w_t})$ be a partition with distinct parts $\mu_1,\ldots,\mu_t$, and let $n\ge k$. Over a field of characteristic zero,
\begin{equation}
m_{\bl}(x_1,\ldots,x_n)
=
\frac{1}{\kappa(\bl)}
\sum_{\mathcal{I}\in\mathcal{P}_k}
(-1)^{k-|\mathcal{I}|}
\prod_{I\in\mathcal{I}}\bigl((|I|-1)!\,p_{\lambda_I}(x_1,\ldots,x_n)\bigr),
\end{equation}
where $\kappa(\bl)=\prod_{j=1}^t w_j!$ and $\lambda_I=\sum_{i\in I}\lambda_i$.
\end{corollary}

\subsection{Tools from Additive Combinatorics.}

To extract the perfect matching polynomial of complete graphs from the power-sum representation of $m_{\bl}$, we also need the notion of Sidon sets. A set of positive integers $S$ is called a \emph{Sidon set} if, for any four elements $a,b,c,d \in S$, the equality
\[
a+b = c+d
\]
implies that the unordered pairs $\{a,b\}$ and $\{c,d\}$ are equal. In other words, unordered pairwise sums are distinct. We will use the following classical result on the existence of large Sidon subsets of finite sets of positive integers.

\begin{theorem}[Existence of large Sidon sets~\cite{abbott1990sidon}]
\label{thm:sidonexistence}
For every finite set $A \subseteq \NN$ of sufficiently large cardinality, there exists a Sidon subset $S \subseteq A$ such that
\[
\card{S} \ge 0.08 \sqrt{\card{A}}.
\]
\end{theorem}

We shall need a slightly stronger version of the above theorem, in the sense that we want Sidon subsets $S \subseteq A$ such that the sum of two elements of $S$ is not equal to any element of $S$. This stronger version (\Cref{cor:sidon-sum-free-subset}) follows from the following result of Erd\H{o}s \cite{erdos1965extremal}.

\begin{theorem}[Erd\H{o}s,  {\cite[Theorem 2]{erdos1965extremal}}]
\label{thm:erdos-sum-free-subset}
For every finite set $A$ of non-zero real numbers, there exists a subset $S \subseteq A$ with $\card{S} \ge \card{A}/3$ such that
\[
a+b \neq c \quad \text{for any } a,b,c \in S.
\]
\end{theorem}
Sets $S$ as above are called \emph{sum-free}. The above theorem states that every finite set of non-zero real numbers contains a large sum-free subset. The following corollary is an immediate consequence of \Cref{thm:sidonexistence} and \Cref{thm:erdos-sum-free-subset}.
\begin{corollary}
\label{cor:sidon-sum-free-subset}
For every finite set $A$ of positive integers of sufficiently large cardinality, there exists a sum-free Sidon subset $S \subseteq A$ with $\card{S} \ge 0.04 \sqrt{\card{A}}$. More specifically, $S$ satisfies the following two properties:
\begin{align*}
a+b = c+d &\implies \{a,b\} = \{c,d\} \quad \text{for any } a,b,c,d \in S,\\
a+b &\neq c \quad \text{for any } a,b,c \in S.
\end{align*}
\end{corollary}
\begin{proof}
First apply \Cref{thm:erdos-sum-free-subset} to obtain a sum-free subset $A'\subseteq A$ of size at least $\card{A}/3$. If $\card{A}$ is sufficiently large, then so is $\card{A'}$, so we can apply \Cref{thm:sidonexistence} to $A'$ and obtain a Sidon subset $S\subseteq A'$ with
\[
\card{S}\ge 0.08\sqrt{\card{A'}}\ge \frac{0.08}{\sqrt 3}\sqrt{\card{A}}>0.04\sqrt{\card{A}}.
\]
Since $S\subseteq A'$, the sum-free property is preserved.
\end{proof}

\section{Hardness of Blown-Up Matching Polynomials.}
\label{sec:complexity-of-blow-up-matching-polynomial}

In this section, we reduce the multiplicity-free perfect matching polynomial to its blown-up version. This will allow us later to handle $m_{\bl}$ for partitions $\bl$ with repeating parts; the repetitions correspond to blow-ups. The main result of this section is the following lemma.

\begin{lemma}
\label{lem:pm-to-blowup}
Let $r \ge 4$ and $\mathbf{w}=(w_1, \ldots, w_r)\in \ZZ_{>0}^r$ be such that $\sum_i w_i$ is even.
For $t \defeq 2\lfloor r/4 \rfloor$, we have
\[\PMP_t \ \le_c \ \PMP_{r}^{(\mathbf w)}.\]
\end{lemma}

Throughout this section, fix $r \ge 4$ and $t \defeq 2\lfloor r/4 \rfloor$ and $\mathbf w\in \ZZ_{> 0}^r$ as above.
Our reduction proceeds along the following chain, established in \Cref{lem:blowup-perfmatch-match} and \Cref{lem:blowup-match-perfmatch}:
\begin{equation}
\label{eq:blowup-reduction-chain}
\PMP_t
\ \le_c \ \MatchP_{r}^{(\mathbf w)} 
\ \le_c \ \PMP_{r}^{(\mathbf w)}.
\end{equation}

\paragraph{Graph construction.}
Let us write $a_{i}=w_{i}$ and $b_{i}=w_{t+i}$ for $i\in[t]$.
Without loss of generality, we assume that $\mathbf w$ is non-increasing. 
In particular, this implies that $a_{i} \geq b_{i}$ for all $i\in[t]$.
We obtain a weighted gadget $\Gamma_{\mathbf w}$ from $K_r^{(\mathbf w)}$ by making the following substitutions and deleting zero-weight edges:
\begin{itemize}
\item Introduce a new indeterminate $z$ and set $x_{ij} \gets z\cdot x_{ij}$ for all $1\leq i<j \leq t$.
\item Set $x_{i,t+i} \gets  \xi_i$ for $1 \leq i \leq t$, where $\xi_i$ will be determined later. The choice of $\xi_i$ is crucial.
\item Set $x_{i,j}\gets 0$ for all $1 \leq i <j \leq r$ not covered in the previous two steps.
\item Set every intra-block variable $z_i$ of $\MatchP_r^{(\mathbf w)}$ to $0$.
\end{itemize}
Thus $\MPG{\Gamma_{\mathbf w}}$ is a specialization of $\MatchP_r^{(\mathbf w)}$.
Only the first $2t = 4\lfloor r/4 \rfloor$ blocks are relevant; every vertex in the remaining (at most three) blocks is isolated in $\Gamma_{\mathbf w}$.
\Cref{fig:blowup-gadget} shows the case $r=8$ and $t=4$, with the four original blocks having the same sizes and positions as in \Cref{fig:graph-blow-up}.
Note that we \emph{cannot} control the blow-up $\mathbf w$; 
the vector $\mathbf w$ can be viewed as adversarial.  

For $i\in[t]$, let $A_{i}$ be the vertex block in $\Gamma_{\mathbf w}$ corresponding to vertex $i$ of $K_r$,
and let $B_i$ be the block corresponding to vertex $t+i$.
We will call $A_1\cup \ldots \cup A_t$ the ``original'' vertices of $\Gamma_{\mathbf w}$ and $B_1 \cup \ldots \cup B_t$ the ``gadget'' vertices.
Note that each block $A_i$ of original vertices has a private block $B_i$ of gadget vertices; there are no edges between $A_i$ and $B_j$ for $i\neq j$.

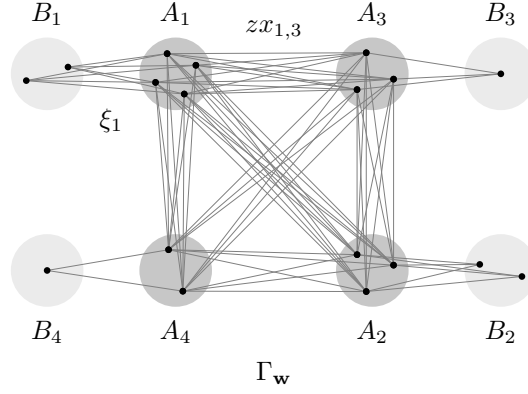
\begin{figure}[t]
\centering
\begin{tikzpicture}[paper figure]
\drawexampleblowup{A}{0}
\foreach \i/\x/\y/\a/\b/\angle in
  {1/-3/1.3/4/2/18,2/3/-1.3/3/2/164,3/3/1.3/3/1/0,4/-3/-1.3/2/1/0} {
  \coordinate (B-\i) at (\x,\y);
  \draw[blob, fill=black!8] (B-\i) circle[radius=.48cm];
  \foreach \j in {1,...,\b} {
    \ifnum\b=1
      \coordinate (B-\i-\j) at (B-\i);
    \else
      \coordinate (B-\i-\j) at
        ($(B-\i)+({\angle+360*(\j-1)/\b}:.29)$);
    \fi
  }
  \foreach \u in {1,...,\a}
    \foreach \v in {1,...,\b}
      \draw[between] (A-\i-\u) -- (B-\i-\v);
  \foreach \j in {1,...,\b}
    \node[copy vertex] at (B-\i-\j) {};
}
\drawexamplevertices{A}
\foreach \i in {1,3}
  \foreach \row in {A,B}
    \node[figure label, anchor=base] at ($(\row-\i)+(0,.7)$) {$\row_\i$};
\foreach \i in {2,4}
  \foreach \row in {A,B}
    \node[figure label, anchor=base] at ($(\row-\i)+(0,-.9)$) {$\row_\i$};
\node[figure label, anchor=base] at ($(A-1)!.5!(B-1)+(0,-.7)$) {$\xi_1$};
\node[figure label, anchor=base] at ($(A-1)!.5!(A-3)+(0,.55)$) {$z x_{1,3}$};
\node at (0,-2.65) {$\Gamma_{\mathbf w}$};
\end{tikzpicture}
\caption{The graph $\Gamma_{\mathbf w}$ in \Cref{lem:blowup-perfmatch-match}, obtained from $K_8^{(\mathbf w)}$ with $\mathbf w=(4,3^2,2^3,1^2)$.
Edges between original blocks $A_i,A_j$ have weight $z x_{i,j}$; edges from $A_i$ to its private block $B_i$ have weight $\xi_i$. Zero-weight edges are not shown.
For the labelled pair with $|A_1|=4$ and $|B_1|=2$, one can take $\xi_1=-1/2$:
the total weight of gadget extensions is $1+8\xi_1+12\xi_1^2=0$ when all vertices of $A_1$ are free,
but $1+6\xi_1+6\xi_1^2=-1/2$ when one vertex of $A_1$ is already matched by an original edge.}
\label{fig:blowup-gadget}
\end{figure}

\paragraph{First step.}
The reduction from $\PMP_t$ to $\MatchP_{r}^{(\mathbf w)}$ proceeds as follows:
Every matching $M$ in $\Gamma_{\mathbf w}$ is the union of an ``original matching'' $O$ that lives within the original blocks and, for each $i\in [t]$, some matching between the original block $A_i$ and its gadget block $B_i$.
In other words, each matching $O$ on the original part can be extended in several ways to a matching of the overall graph $\Gamma_{\mathbf w}$ by adding gadget edges.
For each $i \in [t]$, the edge-weight $\xi_i$ on the gadget edges between $A_i$ and $B_i$ is chosen carefully to ensure that the weighted count of extensions vanishes if $O$ leaves all vertices of $A_i$ unmatched. 
The required weights $\xi_i$ are supplied by the following lemma, which is proven in \Cref{sec:apx-isolation-weight} of the appendix.

\begin{lemma}
\label{lem:isolation-weight}
Let $1\le b\le a$ be integers and let $K_{a,b}$ be the complete bipartite graph with $a+b$ vertices. 
Let $m_k$ for $k\in \NN$ count the matchings with $k$ edges in $K_{a,b}$. Defining the matching-generating polynomial
\[
M(K_{a,b},x) \defeq \sum_{k} m_k \, x^k,
\]
there is $\xi=\xi(a,b) \in \mathbb C$ such that
    \begin{align}
    M(K_{a,b},\xi ) & =0, \label{eq:vanish}\\
    M(K_{a-1,b},\xi ) & \neq0 \label{eq:nonvanish}.
    \end{align}
\end{lemma}

\begin{remark}
For nonuniform $c$-reductions, the mere existence of the constants $\xi(a,b)$ is enough. We remark that they are roots of $M(K_{a,b},x)$, so their algebraic degree is at most $b$ over $\mathbb Q$, and the defining polynomial has integer coefficients whose bit lengths are polynomially bounded in $a$ and $b$.  
\end{remark}

The isolation weights let us recover perfect matchings from the blown-up matching polynomial.
The following lemma makes this reduction precise.

\begin{lemma}
\label{lem:blowup-perfmatch-match}
For $r,t,\mathbf{w}$ as in \Cref{lem:pm-to-blowup}, we have $\PMP_t \le_c \MatchP_{r}^{(\mathbf w)}$.
\end{lemma}
\begin{proof}
Use the graph $\Gamma_{\mathbf w}$ and blocks $A_i,B_i$ defined above. Choose $\xi_i\defeq\xi(a_i,b_i)$ by \Cref{lem:isolation-weight}, and put
\[
\eta_i \defeq M(K_{a_i-1,b_i},\xi_i) \neq 0.
\]
We obtain the required $c$-reduction by proving that
\begin{equation}
\label{eq:pm-main-reduction}
\PMP_t = \frac{[z^{t/2}] \MPG{\Gamma_{\mathbf w}}}{\prod_{i=1}^t a_i \eta_i}.
\end{equation}
To prove \Cref{eq:pm-main-reduction}, let $\mathcal{O}$ be the set of matchings with exactly $t/2$ edges in the graph $\Gamma_{\mathbf w}[A_1 \cup \ldots \cup A_t]$ induced by original blocks. We call these the \emph{original matchings}.
For every original matching $O \in \mathcal O$, define its \emph{type} $\theta_O$ to be the map with domain $[t]$ that maps $i \mapsto |V(O)\cap A_i|$.
An original matching $O \in \mathcal O$ is a \emph{lift} if $\theta_O(i)=1$ for all $i \in [t]$.
For an original edge $e$ between blocks $A_i$ and $A_j$, let $\bar w(e)\defeq x_{ij}$ be its weight after removing the auxiliary factor $z$.

Each lift corresponds to a perfect matching in $K_t$ by consolidating each block $A_i$ of $\Gamma_{\mathbf w}$ to a single vertex $i$ in $K_t$. Conversely, each perfect matching in $K_t$ corresponds to $\prod_{i=1}^t a_i$ lifts in $\Gamma_{\mathbf w}$ by independently choosing, for each block $i$, one of $a_i$ vertices in block $A_i$. The correspondence preserves the weights $\bar w(e)$ of the original edges, so the weighted count of lifts equals that of perfect matchings in $K_t$, up to the factor $\prod_{i=1}^t a_i$.

In view of this, we aim at restricting the matchings counted in $[z^{t/2}] \MPG{\Gamma_{\mathbf w}}$ to lifts.
For $O \in \mathcal O$, let $\mathcal E_O$ be the set of matchings on gadget edges that extend $O$ to a matching $M$ in $\Gamma_{\mathbf w}$.
Every matching of $z$-degree $t/2$ has a unique decomposition $O\cup F$, where $O\in\mathcal O$ and $F\in\mathcal E_O$. Hence
\begin{equation}
\label{eq:match-q}
[z^{t/2}] \MPG{\Gamma_{\mathbf w}} = \sum_{O \in \mathcal O} \prod_{e\in O} \bar w(e) \underbrace{\sum_{F\in\mathcal E_O} \prod_{e\in F} w(e)}_{=q(O)}.
\end{equation}
We show that the gadget extensions $q(O)$ restrict the summation to lifts $O$. 
To abbreviate, let us write $\theta \defeq \theta_O$ and write $Q_i^\theta$ for the complete bipartite graph with $a_i-\theta(i)$ left and $b_i$ right vertices.
We observe that
\begin{equation}
\label{eq:qO-prod}
q(O) 
= \prod_{i=1}^t M(Q_i^{\theta} ,\xi_i)
=\begin{cases}
0 & O \text{ is not a lift,}\\
\prod_{i=1}^t \eta_i & O \text{ is a lift.}\\
\end{cases}
\end{equation}
For the first equation, note that the sum defining $q(O)$ in \Cref{eq:match-q} ranges over all matchings in the disjoint union of complete bipartite graphs $Q_1^{\theta} \cup \ldots \cup Q_t^{\theta}$, where $Q_i^{\theta}$ has weight $\xi_i$ on all edges.
For the second equation:
\begin{itemize}
    \item If $O$ is not a lift, we show that $q(O)=0$. Indeed, if $O$ is not a lift, then $\theta(i)=0$ for some $i\in [t]$, because $O$ is a matching with $t/2$ edges by definition, so an index $j$ with $\theta(j)>1$ forces some index $i$ with $\theta(i)=0$. Then $Q_i^{\theta} \cong K_{a_i,b_i}$ and our choice of $\xi_i$ by \Cref{lem:isolation-weight} ensures that the $i$-th factor in \Cref{eq:qO-prod} is $M(Q_i^{\theta},\xi_i)=0$, which overall implies $q(O)=0$.
    \item If $O$ is a lift, then $Q_i^{\theta} \cong K_{a_i-1,b_i}$ for all $i$, so the factors in \Cref{eq:qO-prod} for $i \in [t]$ are $M(Q_i^{\theta},\xi_i) = \eta_i \neq 0$.
\end{itemize}
Substituting \Cref{eq:qO-prod} into \Cref{eq:match-q}, we obtain
\begin{equation}
[z^{t/2}] \MPG{\Gamma_{\mathbf w}}
\ = \  \left(\prod_{i=1}^t \eta_i\right) \sum_{O \in \mathcal O \text{ lift }} \prod_{e\in O} \bar w(e)
\ = \ \left(\prod_{i=1}^t a_i\eta_i \right) \PMP_t,
\end{equation}
where the second equality follows from an earlier observation: every perfect matching in $K_t$ corresponds to $\prod_i a_i$ lifts $O \in \mathcal O$ of the same $\bar w$-weight. This proves \Cref{eq:pm-main-reduction} and concludes the proof.
\end{proof}

\paragraph{Second step.}
We reduce the blown-up matching polynomial to the blown-up perfect matching polynomial.
\begin{lemma}
\label{lem:blowup-match-perfmatch}
For $r,\mathbf{w}$ as in \Cref{lem:pm-to-blowup}, we have $\MatchP_r^{(\mathbf w)} \le_c \PMP_r^{(\mathbf w)}$.
\end{lemma}
\begin{proof}
Put $n\defeq\sum_{i=1}^r w_i$. Let $\lambda$ be a new variable and let $q$ be the polynomial obtained from $\PMP_r^{(\mathbf w)}$ by substituting $\theta \gets \lambda \theta +1$ for every edge variable $\theta$ of the blow-up graph, including the inter-block variables $x_{ij}$ and the intra-block variables $z_i$.
To prove the lemma, we establish the identity
\begin{equation}
\label{eq:match-from-perfmatch}
\MatchP_r^{(\mathbf w)}
= 
\sum_{k=0}^{n/2}
\frac{[\lambda^{k}] q}{(n-2k-1)!!},
\end{equation}
where, for even $m$, $(m-1)!!=(m-1)(m-3)\cdots 1$ is the number of perfect matchings in $K_m$, with the convention $(-1)!!=1$.
The reduction follows by summing over $k$ and performing interpolation on $\lambda$, similar to \Cref{lem:homogeneous-component-complexity}.

To prove \Cref{eq:match-from-perfmatch}, let $Q_r^{(\mathbf w)}$ be the multigraph with, for each pair of vertices $u,v$, one ``original'' edge of weight $\lambda\theta_{uv}$ and one ``neutral'' edge of weight $1$, where $\theta_{uv}$ is the corresponding inter-block or intra-block variable.

Every $k$-matching $M$ in $K_r^{(\mathbf w)}$ can be extended to exactly $(n-2k-1)!!$ perfect matchings in $Q_r^{(\mathbf w)}$ by adding a perfect matching on neutral edges on the $n-2k$ vertices not matched by $M$.
Each such extension has the same weight as $M$, up to a factor $\lambda^k$.
The coefficient of $\lambda^k$ in $q$ sums over all such extensions of $k$-matchings in $K_r^{(\mathbf w)}$.
\end{proof}

We have collected both steps required to prove \Cref{lem:pm-to-blowup}.

\begin{proof}[Proof of \Cref{lem:pm-to-blowup}]
Combine \Cref{lem:blowup-perfmatch-match} and \Cref{lem:blowup-match-perfmatch} with the transitivity of $c$-reductions via \cref{eq:blowup-reduction-chain}.
\end{proof}

\section{Blown-Up Matching Polynomial to Monomial Symmetric Polynomials.} \label{sec:blow-up-matching-polynomial-to-monomial-symmetric-polynomials}

In this section, we extend ideas from \cite{curticapeanmonomialmymmetricpolynomials2022} to prove $\VNP$-completeness of a large class of monomial symmetric polynomial families under $c$-reductions:
if the defining partitions contain sufficiently many distinct parts and the number of variables is at least the degree, then the corresponding family of monomial symmetric polynomials is $\VNP$-complete under $c$-reductions.

For a partition $\bl$ with degree $d(\bl)$, define $m_{\bl}^{\pow}$ to be the unique polynomial in the power-sum variables $y_1,\ldots,y_{d(\bl)}$ such that, when the number of variables is at least $d(\bl)$,
\[
m_{\bl}(x_1,\ldots,x_v)=m_{\bl}^{\pow}(p_1(x),\ldots,p_{d(\bl)}(x)).
\]

We now state the key lemma underlying our main result.

\begin{lemma}\label{lem:reduction-to-blow-up-matchings-to-mbl}
There is an absolute constant $\alpha>0$ such that the following holds for all sufficiently large $t$.
Suppose $\bl=(\lambda_1,\lambda_2,\ldots,\lambda_k)$ is a partition with $D(\bl)=t$. Then there exist $r>\alpha\sqrt{t}$ and a vector $\mathbf{w}\in\ZZ_{>0}^r$ with even $\norm{\mathbf{w}}_1$ such that
\[
\PMP_r^{(\mathbf{w})} \le_c m_{\bl}(x_1,\ldots,x_v)
\]
for all $v \ge d(\bl)$.
\end{lemma}

\begin{figure}[t]
\centering
\begin{tikzpicture}[paper figure]
\foreach \x/\mu/\w in {0/18/4,1/14/3,2/12/3,3/11/2}
  \foreach \j in {1,...,\w}
    \node[part token] at ({\x+.11*(\w-\j)},{.65+.10*(\w-\j)})
      {\ifnum\j=\w\relax$\mu$\fi};
\foreach \x/\i in {0/1,1/2,2/3,3/4}
  \node[inner sep=1pt] at (\x,.1) {$\mu_\i$};
\foreach \x/\mu/\w in {4.5/5/2,5.5/3/1}
  \foreach \j in {1,...,\w}
    \node[part token, fill=black!8] at ({\x+.11*(\w-\j)},{.65+.10*(\w-\j)})
      {\ifnum\j=\w\relax$\mu$\fi};
\foreach \x/\i in {4.5/5,5.5/6}
  \node[inner sep=1pt] at (\x,.1) {$\mu_\i$};
\draw[decorate, decoration={brace, mirror, amplitude=4pt}, line width=.5pt]
  (-.3,-.25) -- (3.45,-.25);
\node[anchor=north, align=center, inner sep=1pt] at (1.575,-.55)
  {parts from $S$ \\ with repetitions};
\draw[decorate, decoration={brace, mirror, amplitude=4pt}, line width=.5pt]
  (4.2,-.25) -- (5.85,-.25);
\node[anchor=north, align=center, inner sep=1pt] at (5.025,-.55)
  {parts not in $S$};
\node[anchor=base] at (2.775,-2.3) {$m_{\bl}$ for $\bl=(18^4,14^3,12^3,11^2,5^2,3)$};
\node[inner sep=1pt] at (8.8,1.5) {extract degree $9$ from $m_{\bl}^{\pow}$};
\node[inner sep=1pt] at (8.8,1) {set $y_\mu=0$ for $\mu\in S$};
\draw[-{Stealth[length=2.5mm]}, thick] (6.75,.25) -- (10.85,.25);
\node[inner sep=1pt] at (8.8,-.55) {divide by $y_5^2\,y_3$};
\begin{scope}[xshift=13.5cm]
  \drawexampleblowup{H}{1}
  \node[anchor=base] at (0,-2.3) {$\mathrm{PerfMatch}(K_4^{(4,3,3,2)})$};
\end{scope}
\end{tikzpicture}
\caption{The extraction in \Cref{lem:reduction-to-blow-up-matchings-to-mbl} for a concrete
$\bl$ and $S=\{18,14,12,11\}$.
Extracting degree $9$ from $m_{\bl}^{\pow}$ and setting $y_\mu=0$ for $\mu\in S$ forces the three outside positions to be singletons and the selected copies to be paired; then dividing by $y_5^2y_3$ removes the forced singleton factors.
Pairs of selected parts are uniquely decodable from their sums.
The resulting polynomial is a nonzero scalar multiple of $\mathrm{PerfMatch}(K_4^{(4,3,3,2)})$, where blowup sizes correspond to repetition counts.}
\label{fig:partition-lemma}
\end{figure}
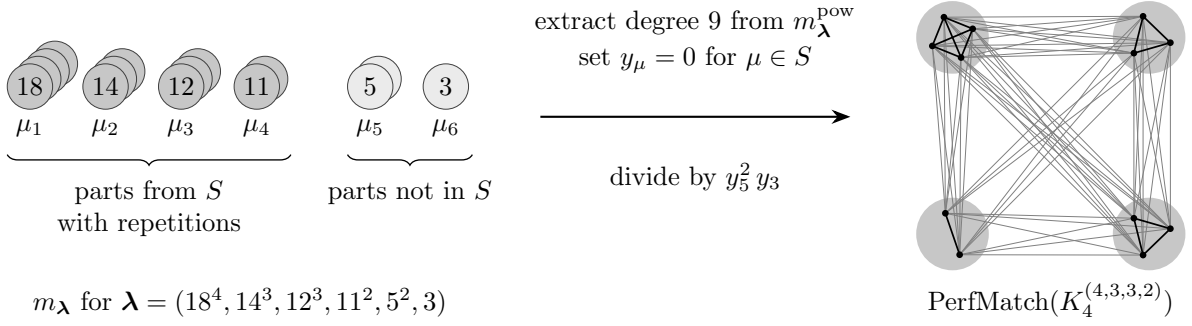

\begin{proof}
Write the partition $\bl$ as
$
\bl = (\mu_1^{w_1}, \mu_2^{w_2}, \ldots, \mu_t^{w_t}),
$
where $\mu_1 > \mu_2 > \cdots > \mu_t > 0$ are the distinct parts of $\bl$ and $w_i$ denotes the multiplicity of $\mu_i$ for each $i \in [t]$.
By \Cref{cor:sidon-sum-free-subset}, there exists a sum-free Sidon subset
$
S_0 \subseteq \{\mu_1, \mu_2, \ldots, \mu_t\}
$
of size at least $0.04\sqrt{t}$. If necessary, remove one element from $S_0$ so that the total multiplicity of the remaining set is even, and call the resulting set $S$. For all sufficiently large $t$, this leaves $r\defeq |S|> \alpha\sqrt{t}$ for a fixed constant $\alpha>0$. Suppose we have
\[
S= \{\mu_{i_1}, \mu_{i_2}, \ldots, \mu_{i_r}\} \text{ for some } 1 \le i_1 < i_2 < \cdots < i_r \le t. 
\]
$
W \defeq w_{i_1}+\cdots+w_{i_r}
$
is even. By \Cref{lem:pow-reduction}, we have
\begin{equation}
\label{eq:pow-reduction-complexity}
L^{m_{\bl}(x_1,\ldots,x_v)}(m_{\bl}^{\pow})
\le
\tilde{O}(v^4 (d(\bl))^3).
\end{equation}
For brevity, let $\mathbf{y} = \{y_1, y_2, \ldots, y_{d(\bl)}\}$. By definition of $m_{\bl}^{\pow}$ and \Cref{cor:binet-minc}, we have:
\[
m_{\bl}^{\pow}(\mathbf{y})
=
\frac{1}{\kappa(\bl)}
\sum_{\mathcal{I} \in \mathcal{P}_k}
(-1)^{k-|\mathcal{I}|}
\prod_{I \in \mathcal{I}} (|I|-1)!
\, y_{\lambda_I},
\]
where $\lambda_I = \sum_{i \in I} \lambda_i$ for any nonempty $I \subseteq [k]$. We consider the homogeneous component of degree $k - W/2$ in $m_{\bl}^{\pow}$ and we substitute $y_{\mu_{i_1}} = y_{\mu_{i_2}} = \cdots = y_{\mu_{i_r}} = 0$ in this homogeneous component to obtain a new polynomial $g(\mathbf{y})$. Note that since $S$ is a sum-free Sidon set, this substitution cannot make any variable of the form $y_{\mu_{i_j} + \mu_{i_\ell}}$ vanish in $g(\mathbf{y})$. By \Cref{lem:homogeneous-component-complexity} we obtain:
\begin{equation}
\label{eq:complexity-upper-bound-for-g}
L^{m_{\bl}^{\pow}}\paren{g(\mathbf{y})}
\le
O(d(\bl)^2).
\end{equation}
Each summand in the above expansion of $m_{\bl}^{\pow}$ corresponds to a partition $\mathcal{I}$ of the labelled positions $[k]$, and the degree of this summand is exactly $\card{\mathcal{I}}$. Thus a partition $\mathcal{I}$ contributes to $g(\mathbf{y})$ only if $|\mathcal{I}| = k - W/2$. Such a partition must contain at least $k - W$ singleton parts. Since we set
\[
y_{\mu_{i_1}} = y_{\mu_{i_2}} = \cdots = y_{\mu_{i_r}} = 0,
\]
singletons corresponding to $\mu_{i_1},\ldots,\mu_{i_r}$ do not contribute. Hence every contributing partition $\mathcal{I}$ must contain exactly $k - W$ singleton parts corresponding to $\mu_j$ with $j\in [t]\setminus \{i_1, i_2,\ldots, i_r\}$.

After these $k-W$ forced singletons are fixed, the remaining $W$ elements are arranged into $W/2$ non-singleton parts.
Since every non-singleton part has size at least two, all of these parts have size exactly two.
Thus every remaining part pairs two elements from $S = \{\mu_{i_1},\ldots,\mu_{i_r}\}$, with $\mu_{i_j}$ appearing exactly $w_{i_j}$ times across these pairs.
Because $S = \{\mu_{i_1},\ldots,\mu_{i_r}\}$ is a sum-free Sidon set, all sums
\[
\mu_{i_j} + \mu_{i_\ell} \quad (1 \le j \le \ell \le r)
\]
are distinct. Moreover, since $S$ is sum-free, none of these sums is equal to any element of $S$, so the substitution $y_{\mu_{i_1}}=\cdots=y_{\mu_{i_r}}=0$ does not delete a pair variable.
It might still happen that $\mu_{i_j} + \mu_{i_\ell} = \mu_q$ for some $q \in [t]\setminus \{i_1, i_2, \ldots, i_r\}$. But all copies of $\mu_q$ outside $S$ occur as forced singletons, hence
\[
\prod_{j\in [t]\setminus \{i_1, i_2, \ldots, i_r\}} y_{\mu_j}^{w_j}
\quad\text{divides}\quad
g(\mathbf{y}).
\]
Hence, using \Cref{lem:divide-by-monomial}, we can divide $g(\mathbf{y})$ by this monomial to obtain a new polynomial $h(\mathbf{y})$ such that:
\begin{equation}\label{eq:bound-after-monomial-division}
  L^{m_{\bl}^{\pow}}(h(\mathbf{y})) \leq \poly(d(\bl)).
\end{equation}
Now, in $h(\mathbf{y})$, we perform the following substitutions for variables that occur:
\begin{itemize}
  \item Each variable $y_{\mu_{i_j}+\mu_{i_k}}$ is substituted by a variable $x_{j,k}$ for $1 \le j < k \le r$. 
  \item Each variable $y_{2\mu_{i_j}}$ is substituted by a new variable $z_j$ for all $j \in [r]$.
\end{itemize}
This yields a polynomial $H$ in the variables $x_{i,j}$ for $1 \le i < j \le r$ and $z_j$ for $j \in [r]$. From \Cref{eq:bound-after-monomial-division} it follows immediately that
\[
L^{m_{\bl}^{\pow}}(H)
\le
\poly(d(\bl)).
\]

Combining \Cref{eq:pow-reduction-complexity}, \Cref{eq:bound-after-monomial-division}, and \Cref{lemma:c-reduction-composition}, we obtain
\[
L^{m_{\bl}(x_1,\ldots,x_v)}(H)
\le
\poly(v,d(\bl)).
\]

Finally, observe that each surviving summand defining $H$ corresponds to a partition $\mathcal{I}$ of the labelled positions $[k]$ with exactly $W/2$ parts of size two arising from elements in $S$ with $\mu_{i_j}$ appearing $w_{i_j}$ many times. Hence each such partition corresponds to a perfect matching of the blow-up graph $K_r^{(\mathbf{w})}$, and vice versa, by identifying the labelled positions carrying $\mu_{i_j}$ with the vertices of block $j$. This correspondence preserves monomial weights, and all surviving summands have the same nonzero coefficient because their partitions have the same number of parts, all of size one or two. See \Cref{fig:partition-lemma} for a visualization. Therefore the polynomial $H$ is a nonzero scalar multiple of $\PMP_r^{(\mathbf{w})}$, where $\mathbf{w} = (w_{i_1}, w_{i_2}, \ldots, w_{i_r})$. Since we work over $\CC$, this scalar can be inverted in the reduction. Hence $L^{m_{\bl}(x_1,\ldots,x_v)}(\PMP_r^{(\mathbf{w})})$ is polynomially bounded in $v$ and $d(\bl)$, which proves the lemma.

\end{proof}

To prove our main result claimed in \Cref{thm:main-vnp-hardness}, we use the results proved in \Cref{sec:complexity-of-blow-up-matching-polynomial}.

\begin{proof}[Proof of \Cref{thm:main-vnp-hardness}]
Let $(f_n)_{n \in \NN}$ be a $p$-family of monomial symmetric polynomials. Write
\[
f_n = m_{\bl_n}(x_1, \ldots, x_{v_n})
\]
for some partition $\bl_n$ and positive integers $v_n$ with $v_n \ge d(\bl_n)$ for all $n \in \NN$.
By assumption, $D(\bl_n) \ge n^\epsilon$ for some constant $\epsilon > 0$.
The family $(f_n)$ is in $\VNP$. Writing $\bl_n=(\mu_1^{w_1},\ldots,\mu_t^{w_t})$, consider the product
\[
\prod_{i=1}^{v_n}\left(1+\sum_{j=1}^t y_j x_i^{\mu_j}\right)
\]
which is a $\VP$ family (and hence a $\VNP$ family), since $v_n$, $t$, and $d(\bl_n)$ are polynomially bounded for a $p$-family. The coefficient of $y_1^{w_1}\cdots y_t^{w_t}$ in this product is exactly $f_n$. By the standard closure of $\VNP$ under taking coefficients of monomials in polynomially many auxiliary variables \cite[Theorem~2.19(3)]{DBLP:books/daglib/0025071}, $(f_n)$ is a $\VNP$ family.

By \Cref{lem:reduction-to-blow-up-matchings-to-mbl}, there exist constants $\alpha > 0$ and $n_0$, and, for each $n\ge n_0$, parameters $r_n > \alpha\sqrt{D(\bl_n)}$ and a vector $\mathbf{w}_n \in \ZZ_{>0}^{r_n}$ with even total sum such that
\[
\PMP_{r_n}^{(\mathbf{w}_n)} \le_c f_n.
\]
Since $D(\bl_n) \ge n^\epsilon$, we have $r_n > \alpha n^{\epsilon/2}$. By \Cref{lem:pm-to-blowup}, setting $s_n = 2\lfloor r_n/4 \rfloor$, we obtain
\[
\PMP_{s_n} \le_c \PMP_{r_n}^{(\mathbf{w}_n)}.
\]
Composing the two reductions via \Cref{lemma:c-reduction-composition} yields
\[
\PMP_{s_n} \le_c f_n.
\]
Thus $s_n \ge c n^{\epsilon/2}$ for some constant $c>0$ and all sufficiently large $n$.
To reduce an arbitrary even-size perfect matching polynomial $\PMP_m$ to the family $(f_n)$, choose an index $N=\max\{n_0,\lceil (m/c)^{2/\epsilon}\rceil\}$.
Then $s_N\ge m$, and $\PMP_m$ is a projection of $\PMP_{s_N}$ by fixing the extra $s_N-m$ vertices to a predetermined perfect matching of weight $1$ and setting all other incident extra edges to $0$.
Since $N$ is polynomially bounded in $m$, the perfect matching family $c$-reduces to $(f_n)$.
The perfect matching polynomial family is $\VNP$-complete under $p$-projections, hence $(f_n)_{n \in \NN}$ is $\VNP$-hard under $c$-reductions. Together with membership in $\VNP$, this proves $\VNP$-completeness.
\end{proof}

\section{Conclusion and Open Problems.}\label{sec:conclusion-and-open-problems}
We have established an almost-complete dichotomy for the algebraic complexity of monomial symmetric polynomials in terms of the number of distinct parts in the partition $\bl$. In our proofs, we introduced the blow-up matching polynomial and showed its hardness under arbitrary blow-ups with an even total number of vertices; this may be of independent interest toward robust hardness proofs of perturbed versions of $\VNP$-hard problems.
Our results suggest several directions for future work.
\paragraph{Removal of variable count restriction:} 
All our lower bounds and $\VNP$-hardness results require that the number of variables in the monomial symmetric polynomial is at least its degree. Can we show a lower bound for the case when the number of variables is smaller than the degree? This requirement arises because we rely on the algebraic independence of power symmetric polynomials via \Cref{thm:power-basis}, which only holds when $n \geq d(\bl)$. Removing this constraint would significantly strengthen our results and extend the dichotomy to a broader class of monomial symmetric polynomials.
For a concrete example, consider the Vandermonde permanent
\[
\Per\paren{(x_i^j)_{i,j\in[n]}}
=
m_{(n,n-1,\ldots,1)}(x_1,\ldots,x_n).
\]
This family is not covered by our results, since it has only $n$ variables but degree $\binom{n+1}{2}$. It is a natural open problem whether the Vandermonde permanent family is $\VNP$-complete. 
  
\paragraph{Fine-grained complexity in the middle regime:} What is the complexity of the monomial symmetric polynomial family $m_{\bl}$ when the number of distinct parts in $\bl$ is $\omega(1)$ but lies outside the range covered by our hardness theorem, e.g., $\poly(\log n)$? Can we show a lower bound for this case? Treating the number of distinct parts as a parameter might be useful; parameterized complexity was introduced in algebraic complexity by \cite{parameterizedvaliantclasses2019}.

For immanants, Curticapean \cite{DBLP:conf/stoc/Curticapean21} has shown such a fine-grained result: If a certain parameter is bounded by a constant, then the corresponding immanant family is in $\VP$, whereas unbounded growth places it outside $\VP$, assuming $\mathsf{VFPT} \neq \mathsf{VW}[1]$.   
In a similar spirit, can we for instance show that $m_{\bl}$ is fixed parameter tractable when parameterized by the number of distinct parts in $\bl$, without requiring any restriction on the number of variables? This would provide a more nuanced understanding of the complexity landscape beyond the dichotomy.

\bibliography{ref}
\bibliographystyle{alphaurl}

\appendix
\crefalias{section}{appendix}

\section{Background.}\label{sec:background}

\subsection{Symmetric Polynomials.}
\label{sec:apx-symm-poly}

\begin{theorem}[Fundamental Theorem of Symmetric Polynomials, see {\cite{macdonald1995}}]\label{thm:symmetric-basis}
Every symmetric polynomial in $\FF[x_1, \ldots, x_n]$ can be uniquely expressed as a polynomial in the elementary symmetric polynomials $e_1, e_2, \ldots, e_n$. Namely, for every symmetric polynomial $f \in \FF[x_1, \ldots, x_n]$, there exists a unique polynomial $g \in \FF[y_1, y_2, \ldots, y_n]$ such that
\[
f(x_1, \ldots, x_n) = g(e_1(x_1, \ldots, x_n), \ldots, e_n(x_1, \ldots, x_n)).
\]
\end{theorem}
An alternative set of algebraically independent generators is provided by the \emph{power sum symmetric polynomials}. The following theorem follows from the Newton identities, which relate the power sum symmetric polynomials to the elementary symmetric polynomials; see \cite{macdonald1995}.
\begin{theorem}[Power sum generators]\label{thm:power-basis}
For $n$ variables over $\CC$, the power sum symmetric polynomials
\[
p_1, p_2, \ldots, p_n
\]
are algebraically independent and generate the ring of symmetric polynomials.
Every symmetric polynomial can be uniquely expressed as a polynomial in $p_1, \ldots, p_n$.
Namely, for every symmetric polynomial $f \in \CC[x_1, \ldots, x_n]$, there exists a unique polynomial $g \in \CC[y_1, y_2, \ldots, y_n]$ such that
\[  f(x_1, \ldots, x_n) = g(p_1(x_1, \ldots, x_n), \ldots, p_n(x_1, \ldots, x_n)). \]
\end{theorem}

\subsection{Algebraic Complexity Classes and Reductions.}
\label{sec:apx-act}

In algebraic complexity theory, reductions between polynomials are typically defined via \emph{projections}. A polynomial $f \in \FF[x_1, x_2, \ldots, x_n]$ is said to be a projection of another polynomial $g \in \FF[y_1, y_2, \ldots, y_m]$ if $f$ can be obtained from $g$ by substituting each variable $y_i$ with either a variable $x_j$ or a constant from the field $\FF$. This notion of reduction is quite restrictive, since it does not allow adaptive computations or multiple queries to the polynomial being reduced to.

A more flexible notion is given by $c$-reductions. To define it, we introduce the oracle complexity of a polynomial. The oracle complexity of a polynomial $f$ with respect to another polynomial $g$ is the minimum size of an arithmetic circuit computing $f$ that is allowed to query an oracle for $g$. Following the notation of \cite{curticapeanmonomialmymmetricpolynomials2022}, we denote this quantity by $L^g(f)$. If $L^g(f)$ is polynomially bounded in the number of variables and degree of $f$ and $g$, we simply write $f \le_c g$ to simplify the notation. Throughout the paper, the complexity notation $\tilde{O}$ hides polylogarithmic factors. 

Now we define the notion of $\VNP$-completeness under $c$-reductions.  A sequence $(u) = (u_n)_{n \in \NN}$ of non-negative integers  is called   $p$-bounded if there exists a polynomial function $p$ such that for all $n \in \NN$, we have $u_n \le p(n)$. 

A sequence $(f_n)_{n \in \NN}$ of polynomials is called a $p$-family if the sequences of variable counts and degrees are $p$-bounded.
\begin{definition}
The class $\mathsf{VP}$ consists of all  $p$-families $(f_n)_{n \in \NN}$ such that  $(L(f_n))_{n \in \NN}$ is $p$-bounded.
\end{definition}

\begin{definition}
The class $\mathsf{VNP}$ consists of all $p$-families $(f_n)_{n \in \NN}$ where $f_n \in \FF[x_1, \ldots, x_{v_n}]$ such that there exists a polynomial family $(g_n)_{n \in \NN}$ in $\mathsf{VP}$ with $g_n\in \FF[x_1, \ldots, x_{v_n}, y_1, \ldots, y_{u_n}]$  satisfying:
\[  f_n(x_1, \ldots, x_{v_n}) = \sum_{b \in \{0,1\}^{u_n}} g_n(x_1, \ldots, x_{v_n}, b_1, \ldots, b_{u_n}). \]
\end{definition}
A $p$-family $(f_n)_{n \in \NN}$ is $c$-reducible to $(g_n)_{n \in \NN}$, denoted by $(f_n) \preceq_c (g_n)$, if there exists a  $p$-bounded sequence $(t_n)_{n \in \NN}$ such that $L^{g_{t_n}}(f_n)$ is $p$-bounded. It is well known that the permanent polynomial family $(\Per_n)_{n \in \NN}$ defined as \[\Per_n(X) = \sum_{\sigma \in S_n} \prod_{i=1}^n x_{i,\sigma(i)},\] is $\VNP$-complete under $c$-reductions. In fact, $(\Per_n)_{n \in \NN}$ is $\VNP$-complete under the more restrictive notion of $p$-projections, which is a special case of $c$-reductions where the reduction is given by a single substitution of variables. The following lemma is an immediate consequence of the definition of $c$-reductions.

\begin{lemma}\label{lemma:c-reduction-composition}
If $f,g,h$ are polynomials such that $L^f(g) \le a$ and $L^g(h) \le b$, then $L^f(h) \le O(ab)$.
\end{lemma}

\subsection{Tools From Algebraic Complexity Theory.}
\label{sec:apx-acttools}

\Cref{cor:binet-minc} expresses the monomial symmetric polynomial $m_{\bl}$ as a polynomial in the power sum symmetric polynomials. Assume that the number of variables $n$ is at least the degree of $m_{\bl}$, i.e., $n \ge d(\bl)$.  Then by \Cref{thm:power-basis}, the power sum symmetric polynomials $p_1, p_2, \ldots, p_n$ are algebraically independent and generate the ring of symmetric polynomials. Theorem~4 of \cite{blasercomplexitysymmetricpolynomials2019}, together with the efficient Newton identities, implies that the arithmetic complexity of $m_{\bl}^{\pow}$ with respect to the oracle $m_{\bl}$ is quite small. This result is stated more explicitly in \cite[Theorem~4.16]{chaugule2023schurpolynomialsdonothavesmallformulasifdeterminantdoesnt}.

\begin{lemma}[\cite{blasercomplexitysymmetricpolynomials2019,chaugule2023schurpolynomialsdonothavesmallformulasifdeterminantdoesnt}]
\label{lem:pow-reduction}
For any partition $\bl$ and all $n \ge d(\bl)$, we have
\[
L^{m_{\bl}(x_1,\ldots,x_n)}\paren{m_{\bl}^{\pow}} \le \tilde{O}(n^4 (d(\bl))^3).
\]
\end{lemma}

These results prove a more general statement: for any symmetric polynomial, the complexity of expressing it as a polynomial in the power sums is small relative to an oracle for the original polynomial. The above lemma is the special case corresponding to monomial symmetric polynomials.

We will also need the following folklore result concerning the arithmetic complexity of homogeneous components of polynomials. For a polynomial $f \in \FF[x_1,\ldots,x_n]$ and a non-negative integer $m$, we denote by $\Hom_m(f)$ the homogeneous component of $f$ of degree $m$, i.e., the sum of all monomials in $f$ that have total degree $m$.

\begin{lemma}[Folklore, Lemma 12 in \cite{blasercomplexitysymmetricpolynomials2019}]
\label{lem:homogeneous-component-complexity}
Let $f \in \FF[x_1,\ldots,x_n]$ be a polynomial of degree $d$. Then for any $m \in \{0,1,\ldots,d\}$, we have:
\[
L^f(\Hom_m(f)) \le O(nd).
\]
\end{lemma}

We also need a result on the arithmetic complexity of dividing polynomials by a monomial factor. This follows immediately from results in \cite{strassen1973vermeidungdivisionen,kaltofen1986uniformclosurepropertiespcomputable, kaltofen1987singlefactorhensellifting, kaltofen1985polynomialtimereductionsmultivariate, kaltofen1989factorizationpolynomialsgivenbystraightlineprograms}, but we give a self-contained proof below for completeness.
\begin{lemma}
\label{lem:divide-by-monomial}
Let $f \in \FF[x_1,\ldots,x_n]$ be a polynomial of degree $d$. Suppose $m$ is a monomial factor of $f$. Then
\[
L^f(f/m) \le \poly(n,d).
\]
\end{lemma}
\begin{proof}
Suppose $f\in \FF[x_1,\ldots,x_n]$ is a polynomial of degree $d$. Suppose $m$ is a monomial factor of $f$. Then $f = g m$, where $g$ is a polynomial of degree at most $d$. Define $f_{\mathrm{shift}}(x)\defeq f(x_1+1,\ldots,x_n+1)$ and $g_{\mathrm{shift}}(x)\defeq g(x_1+1,\ldots,x_n+1)$.
The oracle for $f$ computes $f_{\mathrm{shift}}$ by querying $f$ at the shifted inputs. This substitution does not change the degree of the polynomial, so $f_{\mathrm{shift}}$ is still a polynomial of degree $d$. Moreover, under this substitution, the monomial $m$ is mapped to a polynomial of the form $1 - h$ with constant term of $h$ being zero. The polynomial $h$ is explicitly computable from the monomial $m$, and hence so are its powers. Thus we can write $f_{\mathrm{shift}} = g_{\mathrm{shift}} (1-h)$. Now we have $g_{\mathrm{shift}}= f_{\mathrm{shift}}/(1-h) = f_{\mathrm{shift}}(1+h+h^2+\cdots)$.  We use $F$ to denote the polynomial $f_{\mathrm{shift}}(1+h+h^2+\cdots+h^d)$. Note that $g_{\mathrm{shift}}$ is a polynomial of degree at most $d$, hence we have:
\[g_{\mathrm{shift}} = \sum_{i=0}^d \Hom_i(F).\]
Indeed, $F=g_{\mathrm{shift}}(1-h^{d+1})$, and the term $g_{\mathrm{shift}}h^{d+1}$ has no homogeneous component of degree at most $d$ because $h$ has zero constant term. Since $\deg(F)=O(d^2)$, the sum of the homogeneous components of $F$ of degree at most $d$ can be recovered by the usual interpolation in an auxiliary scaling variable using polynomially many evaluations of $F$. Finally substituting $x_i\gets x_i-1$ recovers $g$ from $g_{\mathrm{shift}}$. Thus $L^f(g) \le \poly(n,d)$.
\end{proof}

\section{Omitted Proofs.}\label{sec:missing-proofs}

\begin{proof}[Proof sketch of \Cref{thm:easy-case}]
Let $\bl=(\mu_1^{w_1}, \mu_2^{w_2}, \ldots, \mu_t^{w_t})$ be a partition with at most $t$ distinct parts. Then the monomial symmetric polynomial $m_{\bl}$ can be written as the coefficient of $y_1^{w_1} y_2^{w_2} \cdots y_t^{w_t}$ in the polynomial $G$ defined below:
\[
G(x_1, x_2, \ldots, x_n, y_1, y_2, \ldots, y_t) \defeq \prod_{i=1}^n \paren{1+\sum_{j=1}^t y_jx_i^{\mu_j}}.
\]
Note that the individual degree of each $y_j$ in $G$ is at most $n$. We now think of $G$ as a polynomial in the variables $y_1, y_2, \ldots, y_t$ with coefficients in $\FF[x_1, x_2, \ldots, x_n]$. We use the well-known Kronecker substitution and replace $y_i$ by $y^{(n+1)^{i-1}}$ for all $i\in[t]$ to obtain a univariate polynomial $H(y)$ in $y$ with coefficients in $\FF[x_1, x_2, \ldots, x_n]$. This mapping ensures that each monomial $y_1^{a_1} y_2^{a_2} \cdots y_t^{a_t}$  of $G$ is mapped to a unique monomial $y^{a_1 + a_2(n+1) + \cdots + a_t(n+1)^{t-1}}$ of $H(y)$. Define $W\defeq w_1 + w_2(n+1) + \cdots + w_t(n+1)^{t-1}$.

In particular, the monomial $y_1^{w_1} y_2^{w_2} \cdots y_t^{w_t}$ is mapped to $y^{W}$. Hence, we can extract the coefficient of $y^{W}$ in $H(y)$ to obtain $m_{\bl}$. Since the degree of $H(y)$ is at most $(n+1)^t-1$, we can evaluate $H(y)$ at $(n+1)^t$ many points and interpolate to extract the coefficient of $y^{W}$. For one evaluation point $\alpha \in \FF$, computing all powers $x_i^{\mu_j}$ with $\mu_j\le d$, forming the $n$ factors, and multiplying them costs $O(nt\cdot \poly(d))$ gates. Since $\bl$ is a partition of degree $d$ into positive parts, we have $t\le d$, so this per-point cost is at most $n\cdot\poly(d)$. The $(n+1)^t=n^{O(t)}$ interpolation points and final linear combination give the claimed $\poly(d)\cdot n^{O(t)}$ bound.
\end{proof}

\section[Proof of the isolation-weight lemma.]{Proof of \Cref{lem:isolation-weight}.}
\label{sec:apx-isolation-weight}

Recall that the matching-generating polynomial of $H$ is
$M(H,x) \defeq \sum_{k} m_k \, x^k$.
Throughout this section, fix $1\le b\le a$ and abbreviate
\begin{align*}
f(x) &=M(K_{a,b},x) \\
g(x) & =M(K_{a-1,b},x).
\end{align*}
By definition,
\[
f(x)=\sum_{k=0}^{b} \binom{a}{k}\binom{b}{k} k! \, x^k,
\]
so $\deg f=b$, $\deg g\le b$, and $f(0)=g(0)=1$. Moreover $f\not\equiv g$, for instance because the coefficient of $x$ is $ab$ in $f$ and $(a-1)b$ in $g$.

We will use the standard relation between matchings in complete bipartite graphs and generalized Laguerre polynomials. For $\alpha>-1$,
\[
L_n^{(\alpha)}(y)\defeq \sum_{k=0}^{n}(-1)^k\binom{n+\alpha}{n-k}\frac{y^k}{k!}.
\]
A direct coefficient comparison gives
\[
M(K_{a,b},x)=b!\,x^b\,L_b^{(a-b)}(-x^{-1}).
\]
By \cite[Section~5.1, (5.1.1)]{szegoorthogonal1975}, $L_n^{(\alpha)}$ is orthogonal on $[0,\infty)$ with weight $e^{-x}x^\alpha$ for $\alpha>-1$. Since $a-b\ge0$, Szeg\H{o}'s zero theorem \cite[Theorem~3.3.1]{szegoorthogonal1975} implies that $L_b^{(a-b)}$ has $b$ distinct positive real zeros. Hence the roots of $f=M(K_{a,b},x)$ are simple and nonzero.

\begin{proof}[Proof of \Cref{lem:isolation-weight}]
The preceding discussion shows that $f$ has $b$ distinct nonzero roots.
If $g$ vanished on all these roots, then $f$ and $g$ would agree on $b+1$ points, namely on $0$ and the $b$ roots.
This would imply $f \equiv g$, because $f$ and $g$ are of degree at most $b$, contradicting $f\not\equiv g$.
It follows that there is a root $\xi$ of $f$ that is not a root of $g$; this is our desired value $\xi =\xi(a,b)$.
\end{proof}

\end{document}